\documentclass[12pt,draftcls,onecolumn]{IEEEtran}
\usepackage{cite}
\usepackage{graphicx}
\usepackage{amsmath}
\usepackage{amssymb}
\usepackage{algorithm}
\usepackage{color}

\begin{document}

\title{Artificial-Noise-Aided Physical Layer Phase Challenge-Response Authentication for Practical OFDM Transmission}

%\author{Xiaofu Wu\IEEEauthorrefmark{1}, \IEEEauthorblockN{Zhen Yang\IEEEauthorrefmark{1} and
%Lu Gan\IEEEauthorrefmark{2}} \\
%\IEEEauthorblockA{\IEEEauthorrefmark{1}
%Nanjing University of Posts and Telecommunications, Nanjing 210003, CHINA\\ Email: xfuwu@ieee.org, and yangz@njupt.edu.cn}
%\\ \IEEEauthorblockA{\IEEEauthorrefmark{2}Brunel University, London UB8 3PH, UK\\ Email: lu.gan@brunel.ac.uk}}

\author{Xiaofu~Wu, Zhen~Yang, Cong~Ling, and Xiang-Gen Xia% <-this % stops a space
\thanks{This  work was supported in part by the National Natural Science Foundation of China under Grants 61372123, 61271335, by the Key University Science Research Project of Jiangsu Province under Grant 14KJA510003.}% <-this % stops a space
\thanks{Xiaofu~Wu and Zhen~Yang are with the Key Lab of Ministry of Education in Broadband Wireless Communication and Sensor Network Technology, Nanjing University of Posts and Telecommunications, Nanjing 210003, China (e-mails:
        xfuwu@ieee.org, yangz@njupt.edu.cn)).}%}
\thanks{Cong Ling is with the Department of Electrical and Electronic Engineering, Imperial College London, London, UK (e-mail: cling@ieee.org).}
\thanks{Xiang-Gen~Xia is with the Department of Electrical and Computer Engineering, University of Delaware, Newark, DE 19716
 (e-mail: xxia@ee.udel.edu).}}

%\markboth{IEEE ITW'2006}{Shell \MakeLowercase{\textit{et al.}}:
%Bare Demo of IEEEtran.cls for Journals}

\maketitle

\begin{abstract}
 Recently, we have developed a PHYsical layer Phase Challenge-Response Authentication Scheme (PHY-PCRAS) for independent multicarrier transmission. In this paper, we make a further step by proposing a novel artificial-noise-aided PHY-PCRAS (ANA-PHY-PCRAS) for practical orthogonal frequency division multiplexing (OFDM) transmission, where the Tikhonov-distributed artificial noise is introduced to interfere with the phase-modulated key for resisting potential key-recovery attacks whenever a static channel between two legitimate users is unfortunately encountered. Then, we address various practical issues for ANA-PHY-PCRAS with OFDM transmission, including correlation among subchannels, imperfect carrier and timing recoveries. Among them, we show that the effect of sampling offset is very significant and a search procedure in the frequency domain should be incorporated for verification. With practical OFDM transmission, the number of uncorrelated subchannels is often not sufficient. Hence, we employ a time-separated approach for allocating enough subchannels and a modified ANA-PHY-PCRAS is proposed to alleviate the discontinuity of channel phase at far-separated time slots. Finally, the key equivocation is derived for the worst case scenario. We conclude that the enhanced security of ANA-PHY-PCRAS comes from the uncertainty of both the wireless channel and introduced artificial noise, compared to the traditional challenge-response authentication scheme implemented at the upper layer.
\end{abstract}


\IEEEpeerreviewmaketitle

\section{Introduction}

\PARstart{E}{nsuring} security of wireless communications has becoming increasingly important. Openness of wireless networks makes them vulnerable to spoofing attacks where an unauthorized user masquerades as another legitimate user. In the past, conventional cryptographic security mechanisms were used to foil such attacks\cite{ProcIEEEAttack}, in which the identity of a user should be authenticated through a challenge-response process, namely, authentication and key agreement (AKA) protocol. The AKA protocol was revised \cite{GSMLee} for stronger security from second-generation (2G) to fourth-generation (4G) systems.  A recent AKA protocol, known as Evolved Packet System AKA (EPS-AKA) \cite{ComMag4G,LTESecurity,Park4G}, has been proposed for the Long Term Evolution (LTE) system. The security of state-of-the-art EPS-AKA protocol comes from computational complexity, namely, the adversary has limited computational power. It is believed that more efforts should be done to prevent potential innovative attacks since the wireless medium offers novel avenues for intrusion.

In recent years, various efforts \cite{DanPHY,YuIFS,YuCOMMAG,VermaAccess,FerIFS,XiaoPHY,TugnaitJSAC,LiangTW,Barcca,BenCOM} have been made in authenticating the transmitter and receiver at the physical layer. In general, these physical layer authentication schemes can be classified as key based or keyless, according to whether a secret key shared between the transmitter and receiver is exploited to authenticate each other or not. In the keyless authentication schemes \cite{XiaoPHY,TugnaitJSAC,LiangTW,Barcca,FerIFS,BenCOM}, some specific features of either the transmitting device or the specific channel between the legitimate users were exploited in order to authenticate the transmission. As an initial trusted transmission is often required for identifying the features, they might be difficult to implement in some practical scenarios. Instead, various key based authentication schemes \cite{DanPHY,YuIFS,YuCOMMAG,VermaAccess} are closer to the traditional challenge-response mechanism, but less prone to attacks due to the protection from the unique randomness of physical characteristics.

For key based challenge-response authentication schemes, two legitimate users, Alice and Bob, shared a secret key. Whenever Alice transmits a random number as the challenge, Bob sends back a response (often called a tag), which is the output of a cryptographic hash function with both the challenge and key as its inputs. By verifying the response with a locally generated tag, Bob's identity can be confirmed. Indeed, both schemes in \cite{DanPHY,YuIFS} follow this authentication mechanism, which are implemented at the physical layer. In \cite{YuIFS}, both Alice and Bob presume public challenges, which are used to generate tags with the shared key, and the tag is physically encapsulated as an embedded fingerprint, which is conveyed with the primary transmission by superposition. The embedded fingerprint is often allocated with low power, which is further corrupted by the channel noise. Hence, its recovery is in general difficult for the adversary, as she/he faces a fundamental information-theoretic challenge, not purely a computational one. The PHYsical layer Challenge-Response Authentication Mechanism (PHY-CRAM) proposed in \cite{DanPHY} implements the conventional challenge-response process at the physical layer, where the randomness of fading channel's amplitude is used to protect both challenge and response (tag). Recently, we proposed a PHYsical layer Phase Challenge-Response Authentication Scheme (PHY-PCRAS) for multicarrier transmission in \cite{WuCL2014}. It requires the channel reciprocity and the randomness of channel-phase response \cite{SecureTrans} for the protection of the shared key from possible eavesdropping.

By exploiting the randomness of physical channels, various physical layer authentication schemes may ensure unconditional security at least for some bits of the shared key (which cannot be broken even if the adversary has unlimited computational power). However, this enhanced security depends heavily on the underlying physical channel, which is often out of our control. In the worst case of static channels (for example, line-of-sight communications), this kind of unconditional security may not be guaranteed. In this paper, we consider to develop an improved version of PHY-PCRAS for practical OFDM transmission, which can guarantee enhanced security even in the worst case of static channels.

The main contributions of this paper are summarized as follows:
\begin{enumerate}
    \item
     We propose a novel artificial-noise-aided PHY-PCRAS (ANA-PHY-PCRAS) for practical OFDM transmission, where the Tikhonov-distributed artificial noise is introduced to interfere with the phase-modulated key for resisting possible attacks. A strictly-positive key equivocation can be ensured even for the worst case scenario.
   \item
     We make a fine improvement on PHY-PCRAS \cite{WuCL2014}, where the estimate of phase differences between subcarriers is simply replaced by the direct estimate of subcarrier phases. This makes the implementation of PHY-PCRAS simpler.
   \item
     A time-separated subchannel allocation scheme is provided to obtain a sufficient number of uncorrelated subchannels. Then, a modified ANA-PHY-PCRAS is proposed for use of time-separated subchannels, which shows its robustness in verification for alleviating the discontinuity of channel phase at far-separated time slots.
    \item
    Various practical issues are discussed with non-ideal OFDM transmission, including imperfect carrier and timing recoveries. In particular, we show that small sampling offsets often result in significant frequency offsets along the allocated subcarriers, which should be compensated for proper verification.
    \item
    We also provide an application model for generating the shared keys between two legitimate nodes in 4G mobile networks. Hence, the conventional challenge-response authentication scheme employed in 4G networks might be replaced by ANA-PHY-PCRAS with enhanced security.
\end{enumerate}

The rest of the paper is organized as follows. In Section II, we propose an ANA-PHY-PCRAS for perfect OFDM transmission, and a time-separated subchannel allocation scheme is presented, along with a modified ANA-PHY-PCRAS. Section-III is devoted to practical issues with non-ideal OFDM transmission. The security analysis of ANA-PHY-PCRAS is given in Sectiion-IV. Simulation results are presented in Section-V, and the conclusion is made in Section-VI.

\newtheorem{lem}{Lemma}
\section{ANA-PHY-PCRAS for Perfect OFDM Transmission}
In this paper, we employ a common Alice-Bob-Eve model, where two trusting parties, Alice and Bob, share some common secrets and they want to authenticate each other, while Eve, as an opponent, has no any knowledge about the shared secrets and wants to impersonate Alice or Bob.

From the viewpoint of modern cryptography, the development of cryptographic primitives should consider the worst case scenario. In the past, various physical layer authentication schemes were proposed and claimed enhanced security of information-theoretic nature, which, however, depends heavily on the randomness of the underlying physical channel. Whenever the physical channel happens to be static, there is simply no guarantee of enhanced security. Therefore, it is essential to consider the worst case of static channels between Alice and Bob for developing physical layer authentication schemes.

\subsection{Basic Idea of ANA-PHY-PCRAS}
We propose a novel ANA-PHY-PCRAS for OFDM transmission, which makes two nontrivial improvements on PHY-PCRAS\cite{WuCL2014}.

Firstly, channel uncertainty has been proved to be essential for ensuring  enhanced security in various physical layer cryptographic approaches. For ANA-PHY-PCRAS, we introduce the Tikhonov-distributed artificial noise to interfere with the phase-modulated key, which could be used to create artificial channel uncertainty. Therefore, the minimum amount of enhanced security of information-theoretic nature can be guaranteed even in the worst case scenario. This contrasts sharply to various reported physical layer authentication schemes, which rely solely on the randomness of the physical channel. Whenever the channel randomness appears, ANA-PHY-PCRAS can be protected by the uncertainty from both the physical channel and artificial noise.

Secondly, we make a fine improvement on PHY-PCRAS, where the estimate of phase differences between subcarriers is simply replaced by the direct estimate of subcarrier phase. It does work as we use a noncoherent metric for verification, which remains unchanged for any random but constant phase increment over all subcarriers.

\subsection{Signal Model for Perfect OFDM Transmission}
In this paper, we assume a multipath fading channel between Alice and Bob. It is often associate with a channel coherence time $T_c$, below which the channel is considered as temporally correlated.

Assuming an OFDM system with $N$ subcarriers, a bandwidth of $W$ Hz and symbol length of $T_f=T_u + T_g$ seconds, of which, $T_g$ seconds are due to the length of cyclic prefix (CP), and $T_u = N/W$. In the following, we use $T_s=T_u/N=1/W$ to denote the sampling period.

The transmitter uses the waveforms
\begin{eqnarray}
  u_k(t) = \left\{\begin{array}{c}  \frac{1}{\sqrt{T_u}} e^{j2\pi\frac{W}{N}k(t-T_g)}, \quad  \text{if} \quad t\in [0,T_f]\\
  0, \quad \quad \quad \quad \quad \quad \quad \quad \text{otherwise}\end{array}\right.
\end{eqnarray}
$k=0, 1, \cdots, N-1$ and the transmitted baseband for an OFDM symbol is
\begin{equation}
   s(t) = \sum_{k=0}^{N-1} x_k u_k(t),
\end{equation}
where $x_k=e^{j \varphi_k}, k=0,1,\cdots,N-1$ are complex numbers from a signal constellation. Since we focus on a phase challenge-response scheme,  $M$-ary PSK modulation is preferred, and hence $\varphi_k \in \Omega \triangleq \left\{0, \frac{2\pi}{M}, \cdots, \frac{2\pi (M-1)}{M} \right\}$.

The signal is transmitted over a frequency-selective fading channel
\begin{equation}
  \label{eq:htau}
  h(\tau,t) = \sum_i \alpha_i(t) \delta(t-\tau_i),
\end{equation}
where $\tau_i$ is the delay of the $i$-th path and $\alpha_i(t)$ is the corresponding complex amplitude.
Assuming the receiver filter is flat within the signal bandwidth, the received signal is
\begin{equation}
   r(t) = \sum_i \alpha_i(t) s(t-\tau_i) + w(t),
\end{equation}
where $w(t)$ is an additive white Gaussian noise process.

Sampling the signal at time instants $t_n=nT_s$ yields
\begin{equation}
   r(t_n) = \sum_i \alpha_i(t_n) s(t_n-\tau_i) + w(nT_s).
\end{equation}
For convenience, assume that the delays $\tau_i$'s are integer multiples of $T_s$. With the sampling period of $T_s=1/W$, the number of resulting samples for each OFDM symbol is $N_f=N + N_g$, where $N_g$ denotes the length of CP. After removing the guard interval and taking the fast Fourier transform (FFT) to the received signal, we get
\begin{equation}
  \label{eq:recSigModel}
  y_k =h_k  x_k + w_k, k=0,1,\cdots, N-1,
\end{equation}
where $y_k = \sum_n r_n e^{-j2\pi\frac{n}{N}k}$
with $r_n = r((n+N_g)T_s)$, and
\begin{eqnarray}
\label{eq:hk}
  h_k \triangleq h_k(t_n) = \sum_i \alpha_i(t_n) e^{-j 2\pi k \frac{\tau_i}{T_u}},
\end{eqnarray}
which keeps constant at least over one OFDM symbol.

Let $f_c$ denote the carrier frequency at the $0$th subcarrier. With perfect OFDM transmission, it can be viewed as parallel multicarrier transmission with a set of carriers $\mathfrak{F}=\{f_c, f_c+\frac{W}{N}, f_c+2\frac{W}{N}, \cdots, f_c+(N-1)\frac{W}{N}\}$.

\subsection{Subchannel Allocation for ANA-PHY-PCRAS}
As a challenge-response process for ANA-PHY-PCRAS, Alice sends a challenge signal to Bob, Bob sends back a response signal, which can be verified by Alice with the shared secret key. With OFDM transmission, $L < N$ subcarriers $\{f_0,f_1,\cdots,f_{L-1}\}\subset \mathfrak{F}$ are selected. We shall show later that the perfect security of ANA-PHY-PCRAS requires independent fading among $L$ carriers. Hence, these carriers should be well separated.

Let $\mathcal{F} = [0,N-1]$ be the set of indexes for $N$ subcarriers in $\mathfrak{F}$. To ensure independence among $L$ subchannels, one has to find a subset of indexes $\Xi=\{l_0,l_1, \cdots, l_{L-1}\} \subset \mathcal{F}$  (of size $L$) with minimum mutual correlation, namely,
\begin{eqnarray}
  \label{eq:mucorr}
  \Xi = \arg \min_{\Xi \subset \mathcal{F}, |\Xi|=L} \max_{l_i \neq l_j \in  \Xi}|\rho_{l_i,l_j}|,
\end{eqnarray}
where
\begin{eqnarray}
  \rho_{l_i,l_j} \triangleq E\left[h_{l_i} h_{l_j}^*\right] \bigg/ {\sqrt{E\left[|h_{l_i}|^2\right] E\left[|h_{l_j}|^2\right]}}
\end{eqnarray}
since $E[h_{l_i}]=0, i\in [0,L-1]$. In practice, the allocated subchannels are often equally spaced, and the value of $\Delta \ell=l_{i+1}-l_i$ determines the minimum mutual correlation.

\subsubsection{Channel model with exponentially decaying power-delay profile}
Consider a time-invariant version of the multipath fading channel model (\ref{eq:htau}), where $\alpha_i$'s are zero-mean complex Gaussian variables with a power delay profile $\theta(\dot{\tau}_i)$ and $\dot{\tau}_i \triangleq \frac{\tau_i}{T_s}$. The normalized delays $\dot{\tau}_i$'s are assumed to be uniformly and independently distributed over the length of CP ($\dot{\tau}_i \in [0,N_g]$), and an exponentially decaying power-delay profile takes the form of $\theta(\dot{\tau}_i)=e^{-\dot{\tau}_i/\dot{\tau}_{\text{rms}}}$. With this channel model, it was shown in \cite{ExpDecay} that the normalized correlation between subcarriers $l_1$ and $l_2$ is a function of frequency separation $\Delta f= (l_2 - l_1)/N$, which takes the form of
\begin{eqnarray}
  \rho_{l_1,l_2} = \frac{1-e^{-N_g \left(\dot{\tau}_{\text{rms}}^{-1}+2\pi j (l_2-l_1)/N\right)}}{\dot{\tau}_{\text{rms}}(1-e^{-N_g\dot{\tau}_{\text{rms}}^{-1}})(\dot{\tau}_{\text{rms}}^{-1}+j2\pi(l_2-l_1)/N)}.
\end{eqnarray}

\newtheorem{thm}{Scenario}
\begin{thm}
\label{s1}
Consider the scenario where the system operates with a bandwidth of $W=20$ MHz, which is divided into $N=2048$ tones with a total symbol period of 108.8 $\mu$s, of which 6.4 $\mu$s constitutes the CP. Hence, $N_g=128$ and $N_f=N+N_g=2176$.
\end{thm}

Let $\sigma_\tau$ be the time delay spread. For the Scenario \ref{s1} with $\sigma_\tau=0.5$ $\mu$s, it gives that $\dot{\tau}_{\text{rms}}=10$, and the frequency-spaced correlation function is plotted in Fig. \ref{fig:fscf}.

\begin{figure}[htb] %[htbp]
   \centering
   \includegraphics[width=0.65\textwidth]{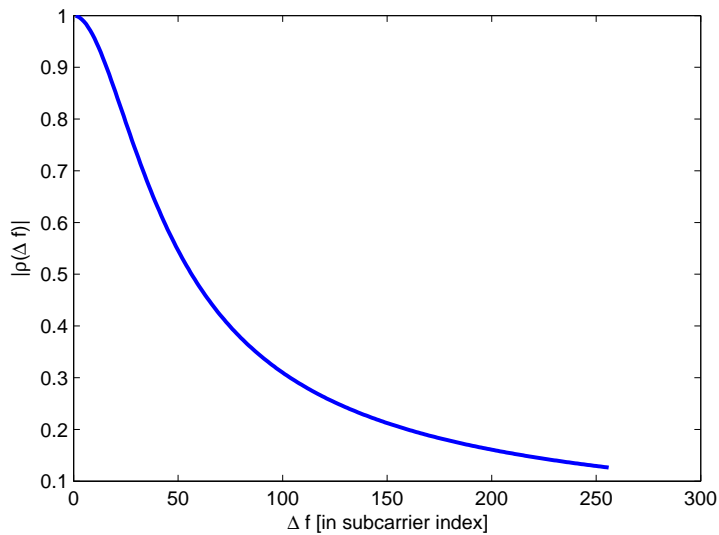}
   \caption{ Frequency-spaced correlation function.}
   \label{fig:fscf}
\end{figure}

%\begin{figure*} %[htb]
\begin{figure}[htb] %[htbp]
   \centering
   \includegraphics[width=0.5\textwidth]{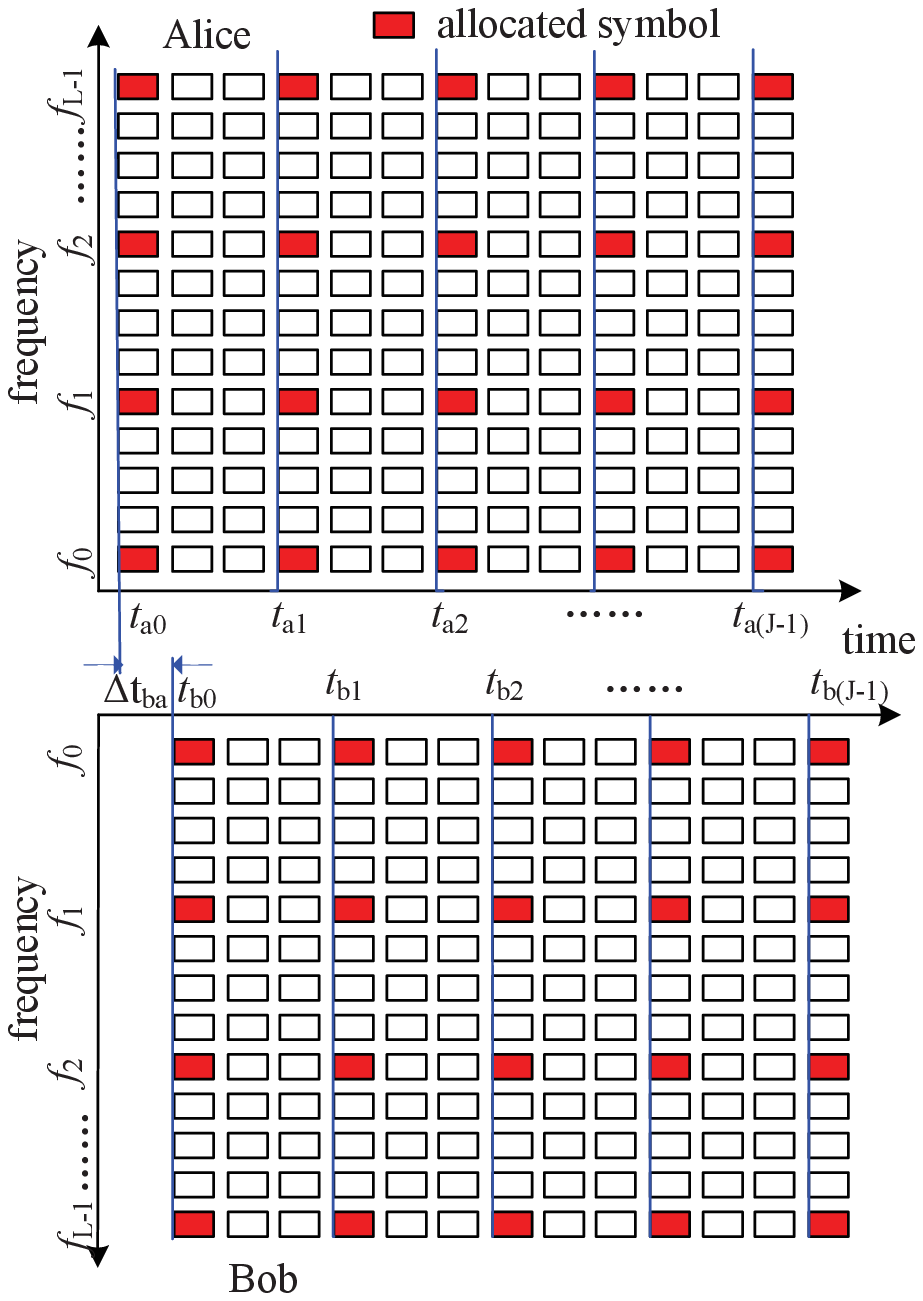}
   \caption{Time-separated allocation of OFDM symbols for PHY-PCRAS.}
   \label{fig:timeFreqAlloc}
\end{figure}
%\end{figure*}
\subsubsection{Time-separated subchannel allocation}
It has been shown that two subchannels could be nearly uncorrelated if they are sufficiently separated, which, however, limits the number of available subchannels for the purpose of physical layer authentication. Consider again the Scenario \ref{s1}. Whenever the allocated subchannels are equally separated with $\Delta \ell =128$, there are only $L'=16+1=17$ well-separated subchannels and the minimum mutual correlation is about 0.2468.

In \cite{WuCL2014}, we have shown that the security of PHY-PCRAS depends on the number of independent subchannels. With BPSK modulation, the size of shared key is equal to the number of independent subchannels. Hence, it is important to allocate much more independent subchannels for use in PHY-PCRAS. Fortunately, one can allocate more subchannels over sufficiently-separated  time slots (OFDM symbols).

The time-separated subchannel allocation scheme is shown in Fig. \ref{fig:timeFreqAlloc}. With sufficiently-separated carriers, there are only $L'$ carriers $f_0, f_1, \cdots, f_{L'-1}$ for use. However, one can repeatedly employ such $L'$ carriers at times $t_0, t_1, \cdots, t_{J-1}$, where $t_j=t_0 + j \cdot \delta T$. To ensure independent fading among different time slots, the minimum time interval between two neighboring time slots should be significantly larger than the channel coherence time, namely, $\delta T >> T_c$.

Coherence time is the time duration over which the channel impulse response is considered to be constant. Channel variation is mainly due to Doppler effects. Using Clarke's model, the coherence time is often selected as $T_c = \sqrt{\frac{9}{16\pi}} f_D^{-1}$, where $f_D$ denotes the maximum Doppler frequency. Consider now that the system operates at carrier frequency of $1.9$ GHz.  In typical urban areas \cite{3GPPCh} with a mobile speed of 50 $km/h$, $f_D \approx 88$ Hz and $T_c \approx 4.8$ ms.

With a challenge-response approach shown in Fig. \ref{fig:codingView}, Alice starts the transmission of challenge signal at time $t_{a0}$, which arrives at Bob later at time $t_{a0} + \delta t$, where $\delta t$ denotes the transmission delay between Alice and Bob. Then, Bob sends back a response signal at time $t_{b0}$. Define $\Delta t_{ba}= t_{b0}-t_{a0}$. Clearly, $\Delta t_{ba} >  \delta t$. \textit{PHY-PCRAS depends on the reciprocity of the channel between Alice and Bob. It is understood that the channel keeps constant during the coherence time $T_c$ and hence the channel reciprocity requires that $\Delta t_{ba} < T_c - T_f$, as shown in Fig. \ref{fig:timeFreqAlloc}}.

\begin{figure*} %[htb]
   \centering
   \includegraphics[width=0.82\textwidth]{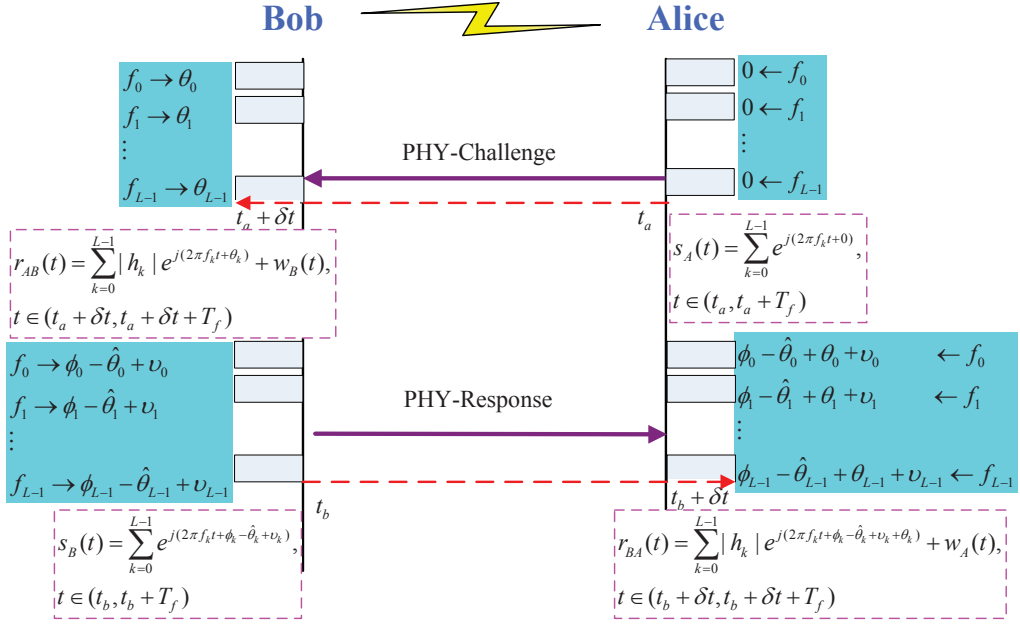}
   \caption{ ANA-PHY-PCRAS for OFDM transmission.}
   \label{fig:codingView}
\end{figure*}

\subsection{ANA-PHY-PCRAS}
For ease of description, we first assume that all the allocated subchannels are from a single OFDM symbol. Later, we shall present a modified ANA-PHY-PCRAS scheme for the time-separated subchannels shown in Fig. \ref{fig:timeFreqAlloc}. In what follows, we suppose that the shared keys between Alice and Bob are denoted as $\{\mathcal{K}_A,\mathcal{K}_B\}$, where each key can be considered as a sequence of random bits.

\subsubsection{PHY-Challenge}

Consider that Alice wants to start a conversation with Bob as shown in Fig. \ref{fig:codingView}.
Alice sends a ``challenge" frame to Bob starting at time instant $t_a$, which is employed by Bob for estimation of channel phases at multiple carriers. Essentially, Alice sends equal-phase modulated sinusoids ($x_k=1, k=0,1,\cdots,L-1$) at frequencies $f_0,f_1,\cdots,f_{L-1}$ during the period of a single OFDM symbol $t\in [t_a, t_a+T_f]$, namely,
\begin{eqnarray}
  \label{eq:cn}
     s_A(t) = \sum_{k=0}^{L-1} e^{j(2\pi f_k t + 0)}, t \in [t_a, t_a+T_f].
\end{eqnarray}

With perfect OFDM transmission, the waveforms $e^{j(2\pi f_k t)}$ can be viewed as ``mutually orthogonal'' \footnotemark\footnotetext{Actually, they are only orthogonal in the discrete time domain, the continuous form is employed to show the time-related issues for convenience.} at the receiver even they undergo multipath fading channels (after insertion and deletion of the CP).
Equivalently, the received signal at Bob can be represented as
\begin{eqnarray}
  \label{eq:cn}
     r_B(t) = \sum_{k=0}^{L-1} |h_k| e^{j(2\pi f_k t + \theta_k)} + w(t), t \in [t_a +\delta t, t_a+\delta t+T_f].
\end{eqnarray}
where $h_k = \sum_i \alpha_i(t) e^{-j 2\pi l_k \frac{\tau_i}{T_u}}, l_k \in \Lambda, k=0, 1, \cdots, L-1$ are assumed to be constant during $t\in [t_a, t_a+\delta t + T_f]$, and $\angle(h_k) = \theta_k$ are channel phase responses at $L$ subcarriers. Hence, a parallel fading channel model $y^B_k =  |h_k| e^{j \theta_k} + w_k, k=0, 1, \cdots, L-1$ is assumed with perfect carrier and timing recoveries (please refer to (\ref{eq:recSigModel})).

Then, Bob estimates the phase at each subcarrier $f_k$, namely,
\begin{eqnarray}
  \label{eq:diff}
     \hat{\theta}_k = \angle (y^B_k) =\theta_k + \Delta \hat{\theta}_k , k=0,1,\cdots,L-1.
\end{eqnarray}
where $\Delta \hat{\theta}_k$ denotes the estimation error. \textit{Noting that we use the absolute channel phase estimates $\hat{\theta}_k$ while the estimates of channel phase differences are employed in PHY-PCRAS \cite{WuCL2014}}. Compared to PHY-PCRAS, the direct estimate of channel phase simplifies the implementation and its robustness against the receiver oscillator remains unchanged as shown later.

\subsubsection{PHY-Response}

At this stage, Bob responds to Alice with a tagged signal, which encapsulates the shared key $\mathcal{K}_B=[\kappa_0,\kappa_1,\cdots,\kappa_{L-1}]^T$ in the form of
\begin{eqnarray}
  \label{eq:res}
     s_B(t) = \sum_{k=0}^{L-1} e^{j(2\pi f_k t + \varphi_k - \hat{\theta}_k + \upsilon_k)},  t \in [t_b, t_b+T_f].
\end{eqnarray}
where $\varphi_k = 2\pi \frac{\kappa_k}{M} \in \Omega, \kappa_k \in \{0,1,\cdots,M-1\}$ since we assume $M$-ary PSK modulation, and $\upsilon_k$ denotes the introduced artificial noise. We assume that $\upsilon_k, k=0,1,\cdots,L-1$ are independent and identically distributed (i.i.d.) with the same probability-density-function (pdf) $f_\upsilon(x)$. Here, we employ the Tikhonov distribution for $f_\upsilon(x)$, namely,
\begin{eqnarray}
   f_\upsilon(x)= \frac{e^{\beta \cos(x)}}{2\pi I_0(\beta)}, x\in (-\pi,\pi].
\end{eqnarray}
where $\beta \ge 0$ determines the dispersion of the distribution, and $I_0(\beta)$ is the modified Bessel function of
the first kind and 0-th order, and $x$ is confined to a support of length $2\pi$ in the vicinity of 0. The use of Tikhonov distributed artificial noise is due to the fact that the Tikhonov distribution maximizes the entropy when the mean and variance of $e^{j\upsilon}$ (or the circular mean and circular variance of $\upsilon$) are specified \cite{Tikohonovbook}.

Then, the received signal at Alice is given by
\begin{eqnarray}
  \label{eq:revA}
     r_A(t) &=&\sum_{k=0}^{L-1}   |h_k| e^{j\left(2\pi f_k t + (\varphi_k - \hat{\theta}_{k} + \upsilon_k) + \theta_k \right)} + w(t)  \nonumber \\
     &=& \sum_{k=0}^{L-1}   |h_k| e^{j \left(2\pi f_k t + \varphi_k -  \Delta\hat{\theta}_k + \upsilon_k \right)} + w(t),
\end{eqnarray}
where $t \in [t_b +\delta t, t_b+\delta t+T_f]$, and $\Delta\hat{\theta}_k=\hat{\theta}_k -\theta_k$.

With perfect carrier and timing recoveries, sampling the signal with frequency $\frac{1}{T_s}$ can obtain $N_f$ samples for each OFDM symbol, as shown in Section-II.B. After removing $N_g$ samples for the guard interval, $N$ samples are transformed using FFT to retrieve $L$ parallel channels (without ISI) at carriers $f_k, k=0, 1, \cdots, L-1$ as
\begin{eqnarray}
   y_k =  \rho_k e^{j \varphi_k} + w_k, k=0, 1, \cdots, L-1
\end{eqnarray}
with $\rho_k = |h_k| e^{j (-\Delta\hat{\theta}_k + \upsilon_k)}$ and $\text{Var}\{w_k\}=\gamma_s^{-1}$.

Hence, the received vector in its complex form can be written as
\begin{eqnarray}
  \label{eq:recv}
     \mathbf{y} = \left[\rho_0 \tilde{\kappa}_0, \rho_1 \tilde{\kappa}_1, \cdots, \rho_{L-1} \tilde{\kappa}_{L-1} \right]^T + \mathbf{w},
\end{eqnarray}
where $\tilde{\kappa}_k=e^{j 2\pi \frac{\kappa_k}{M}}, k=0,1,\cdots,L-1$.

\subsection{Verification}
To complete the authentication process, Alice requires verifying whether the response signal $\mathbf{y}$ is from Bob or not. If the response signal is not from Bob but Eve (an impersonation attacker), it is assumed that Eve generates a length-$L$ $M$-ary random vector $\mathcal{K}_E$ for authentication as there is no information about $\mathcal{K}_B$ available to Eve. Essentially, this is cast as a binary hypothesis testing problem \cite{Maurer}:
\begin{eqnarray}
  \label{eq:cn}
     H_1 &:&  \mathcal{K}_t = \mathcal{K}_B  \nonumber \\
     H_0 &:&  \mathcal{K}_t = \mathcal{K}_E
\end{eqnarray}
where $\mathcal{K}_t$ denotes the acknowledged key.

The optimum binary hypothesis testing was formulated in \cite{WuCL2014}, which is difficult to solve in general. Instead, we propose to use the test statistic
\begin{eqnarray}
  \label{eq:cn}
   \zeta =|\eta|^2,  \eta = \mathcal{K}_B^\dag \mathbf{y},
\end{eqnarray}
where $\mathbf{x}^{\dagger}$ denotes the conjugate transpose of $\mathbf{x}$. Then, $\zeta$ is compared to a threshold value $\iota$ for making a final decision.

In both hypotheses, $\eta$ is the sum of $L$ dependent identically-distributed random variables, which could be approximately regarded as normally distributed for large $L$ from the central limit theorem, especially when the dependence among random variables is weak \footnotemark\footnotetext{The use of i.i.d. artificial noise over time in ANA-PHY-PCRAS makes the dependence among random variables weaker.}. Hence, $\zeta =|\eta|^2$ is noncentrally chi-squared distributed with 2 degrees of freedom, the pdf of which can be expressed as
\begin{eqnarray}
  \label{eq:mod2}
   f_\zeta(x) = \frac{1}{\sigma^2_{H_i}}  e^{-\frac{x+\lambda}{\sigma_{H_i}^2}} I_{0}\left(\frac{2\sqrt{x \lambda}}{\sigma_{H_i}^2}\right),
\end{eqnarray}
where $E\{\zeta\}=\sigma^2_{H_i}+\lambda$ and $\text{Var}\{\zeta\}=2\sigma^2_{H_i}(\frac{1}{2}\sigma^2_{H_i}+\lambda)$ under hypothesis $H_i, i=0,1$. In \cite{RiceEst}, it was shown that $\lambda$ and $\sigma^2_{H_i}$ can be estimated from the moments of $\zeta$ as
\begin{eqnarray}
  \label{eq:riceEst}
   \lambda &=& \sqrt{2 E^2\{\zeta\} - E\{\zeta^2\}}, \nonumber \\
    \sigma_{H_i}^2 &=&  E\{\zeta\} - \lambda .
\end{eqnarray}

\textit{We point out that  the use of $|\mathcal{K}_B^\dag \mathbf{y}|$ for verification makes $\zeta$ unchanged for any random but constant phase rotation among all subcarriers}. \textit{Therefore, the estimate of phase differences $\Delta \theta_{k0}=\theta_k - \theta_0, k=1, \cdots, L-1$ between subcarriers in PHY-PCRAS \cite{WuCL2014} is simply replaced by the direct estimate of subcarrier phases $\theta_k, k=0, 1, \cdots, L-1$}. Even if the receiver oscillator may introduce a random but constant phase rotation among all subcarriers, it does not pose a challenge for practical implementation if there is only one single oscillator in the receiver for all subcarriers.  Furthermore, there is no stringent requirement on a common time reference between users due to the use of noncoherent metric, which is in sharp contrast to the secret generation approach proposed in \cite{INFOCOM2011}.

\subsection{Modified ANA-PHY-PCRAS for Time-Separated Subchannel Allocation}
Consider the time-separated subchannel allocation scheme shown in Fig. \ref{fig:timeFreqAlloc}. With a total of $J$ time slots ($t_m, m=0,\cdots,J-1$), a key can be divided into $J$ sub-keys, namely, $\mathcal{K}_B =[\mathcal{K}_0^T,\cdots,\mathcal{K}_{J-1}^T]^T$, and each sub-key can be delivered through $L'$ carriers.

When Alice challenges at $J$ time instants $t_{am}, m=0,1,\cdots,J-1$ with $L'$ subcarriers for each time instant, Bob extracts $L'$ subcarrier phases at each time instant, and responds to Alice at time instant $t_{bm}$ with a tagged signal containing the $m$-th sub-key $\mathcal{K}_m$. Finally, the received signal at Alice during $t \in [t_{bm}+\delta t, t_{bm}+ \delta t + T_f]$ in a base-band complex vector form can be written as
\begin{eqnarray}
  \label{eq:tr}
     \mathbf{y}(t_{m}) =  e^{j \theta_o (t_{m})} \cdot \left[\rho_0(t_{m}) \tilde{\kappa}_0,  \cdots, \rho_{L-1}(t_{m}) \tilde{\kappa}_{L-1} \right]^T \nonumber  + \mathbf{w}(t_{m}),
\end{eqnarray}
where $\theta_o(t_{m})$ denotes a random but constant phase due to the receiver's oscillator during $t\in [t_{am},t_{bm}+T_f]$, and $\rho_k (t_{m}) = |h_k| e^{j \left[-\Delta\hat{\theta}_k(t_{m}) + \upsilon_k(t_m)\right]}$.

For the robustness of implementation, we always assume that $\theta_o(t_m), m=0, 1, \cdots, J-1$ are independently random variables over $(-\pi,\pi]$, which means that channel phase discontinuity is observed over far-separated time slots. Hence, this discontinuity at different time slots should be seriously considered for verification, and a noncoherent combining method is preferred. Here, we propose a suboptimum hypothesis testing method, which employs a noncoherent combining metric
\begin{eqnarray}
  \label{eq:mod1}
   \zeta = \sum_{m=0}^{J-1} \left|\eta_m\right|^2,  \eta_m = \mathcal{K}_m^\dag \mathbf{y}(t_m).
\end{eqnarray}

With sufficient separation in time, $\eta_m$'s are independent complex Gaussian variables of the same variance. The sum of squares of $J$ independent complex Gaussian variables of the same variance is noncentrally chi-squared distributed with $2J$ degrees of freedom, which yields the pdf of
\begin{eqnarray}
  \label{eq:mod3}
   f_\zeta(x) = \frac{1}{\sigma^2_{H_i}} \left(\frac{x}{\lambda}\right)^{\frac{J-1}{2}} e^{-\frac{x+\lambda}{\sigma_{H_i}^2}} I_{J-1}\left(\frac{2\sqrt{x \lambda}}{\sigma_{H_i}^2}\right),
\end{eqnarray}
where both $\lambda$ and $\sigma^2_{H_i}$ can be again estimated from the moments of $\zeta$ as shown in (\ref{eq:riceEst}).

The cumulative distribution of $\zeta$ can be described by the generalized Marcum Q-function, which is given by
\begin{equation}
F_\zeta(x|H_i)=1-Q_J\left(\frac{\lambda}{\sigma^2_{H_i}},\frac{x}{\sigma^2_{H_i}}\right), i=0,1
\end{equation}
with $Q_J(a,b)=\int_b^{+\infty} \left(\frac{x}{a}\right)^{\frac{J-1}{2}} e^{-(x + a)} I_{J-1}(2\sqrt{ax})dt$.

The authentication is typically claimed if $\zeta\ge \iota$. The threshold $\iota$ of this test is determined for a false acceptance rate (or false alarm probability) $P_f$ according to the distribution of $\zeta|H_0$
\begin{eqnarray}
    \iota = \arg \max_{\iota'}Q_J\left(\frac{\lambda}{\sigma^2_{H_0}},\frac{ \iota'}{\sigma^2_{H_0}}\right) \le  P_f.
\end{eqnarray}
The successful authenticate rate (or detection probability) can be simply computed as
\begin{eqnarray}
        P_D = Q_J\left(\frac{\lambda}{\sigma^2_{H_1}},\frac{ \iota}{\sigma^2_{H_1}}\right).
\end{eqnarray}

Compared to ANA-PHY-PCRAS,  the use of (\ref{eq:mod1}) results in noncoherent combining loss for the modified scheme, which, however, does not require the assumption of phase continuity among different time slots.

\section{Practical Issues with Nonideal OFDM Transmission}
\subsection{Practical Issues}
For a practical OFDM receiver, there is often a local carrier frequency oscillator for demodulation, with which the received radio signal can be converted from radio frequency into baseband. Then, the baseband signal is sampled and discrete-time samples are obtained for subsequent processing, where the sampling clock is derived from a local oscillator. Practically, both timing and carrier references are asynchronous between the transmitter and receiver.  Hence, in a real-world passband transmission system, the following
parameters can cause disturbances in the receiver.
\begin{enumerate}
\item  The carrier frequency oscillator for demodulation at the receiver can be different with the transmitter oscillator, resulting in a carrier frequency offset of $\Delta f$ and a random but constant phase offset of $\Phi_0$.

\item  The sampling time at the receiver has a constant symbol offset $\varepsilon=n_\varepsilon T_s$ compared to the transmitter time.

\item  The sampling time at the receiver has a sampling clock frequency offset of $\varsigma = (T'_s-T_s)/T_s$ compared to the transmitter time, where the sampling period $T'_s$ employed at the receiver is deviated from the desired sampling period $T_s$.
\end{enumerate}

\textit{For simplicity of notation and in order to focus on the pure imperfections at the receiver, we do not include the artificial noise in this section, which, however, is fully considered in simulations}.

\subsection{The Effect of Carrier Frequency Offset}
Whenever the condition 1) occurs, the received samples can be written as
\begin{eqnarray}
  r_n = r((n+N_g)T_s) = \sum_i \alpha_i s(t_n-\tau_i)  e^{j(2\pi n\Delta f T_s + \Phi_0)} = e^{j \Phi_0}\sum_k x_k h_k e^{j2\pi n\frac{k+\vartheta}{N}},
\end{eqnarray}
where  $\vartheta=\Delta f T_u$, and $N_g \Delta f T_s$ is included in $\Phi_0$ for convenience.
As the multipath channel is assumed to be constant during at least one OFDM symbol, we simply use $\alpha_i$ instead of $\alpha_i(t)$ for the $i$th path gain.

After the removal of guard interval from the received samples, the application of FFT yields
\begin{equation}
  \label{eq:rec11}
   y_k = e^{j2\pi(\vartheta\frac{N-1}{2N}+\Phi_0)} \frac{\sin(\pi \vartheta)}{N \sin(\frac{\pi \vartheta}{N})}  h_k  x_k + i_k + w_k,
\end{equation}
where
\begin{equation}
   i_k = e^{j2\pi \Phi_0} \sum_{l \neq k} e^{j2\pi\left((l-k+\vartheta)\frac{N-1}{2N}\right)} \frac{\sin(\pi \vartheta)}{N \sin(\frac{\pi (l-k+\vartheta)}{N})} h_l x_l
\end{equation}
denotes the interchannel interference (ICI).  Due to the use of noncoherent metric (\ref{eq:mod1}) for verification, the extra phase $2\pi(\vartheta\frac{N-1}{2N}+\Phi_0)$ has no impact.

It should be noted that with the presence of carrier frequency offset, the direct loss in SNR  is $-\log10\left(\frac{\sin(\pi \vartheta)}{N \sin(\frac{\pi \vartheta}{N})}\right)$ dB and the frequency offset noise power due to the introduction of ICI $i_k$ can be approximated by \cite{OFDMRev-I}
\begin{equation}
   \sigma_i^2 \approx \frac{\pi^2}{3} (\Delta f T_u)^2
\end{equation}
for the normalized channel gains, namely, $E\left\{|h_k|^2\right\}=1$.

\subsection{The Effect of Sampling Offset}
With a non-zero symbol offset $\varepsilon =n_\varepsilon T_s$, the channel impulse response ``seen'' by the receiver is also shifted in the time scale by $\varepsilon$, which yields
\begin{eqnarray}
  \label{eq:chModelerr}
   h_\varepsilon(\tau,t) = h(\tau-\varepsilon,t-\varepsilon) =\sum_i \alpha_i(t-\varepsilon ) \delta(\tau-\tau_i-\varepsilon )
                         \approx \sum_i \alpha_i(t) \delta(\tau-\tau_i-n_\varepsilon T_s).
\end{eqnarray}
since $\alpha_i(t)$ is assume to be constant during at least one OFDM symbol. Just like in (\ref{eq:hk}), the equivalent channel gain at the $k$th carrier can now be written as
\begin{eqnarray}
  h_k^\varepsilon(t_n) = \sum_i \alpha_i(t_n) e^{-j 2\pi k \frac{\tau_i+n_\varepsilon T_s}{T_u}} = h_k (t_n) e^{-j 2\pi n_\varepsilon k/N}.
\end{eqnarray}

With a time-shift of $n_\varepsilon T_s$, the input samples for demodulation are also shifted by $n_\varepsilon$, which results in both intersymbol interference (ISI) and ICI. The ISI arises since one OFDM symbol window with a nonzero shift $n_\varepsilon \neq 0$ will actually be covered by two OFDM symbols, while ICI is due to the corruption of orthogonality among subcarriers when $n_\varepsilon \neq 0$.
Hence, by neglecting a minor loss ($\frac{N-n_\varepsilon}{N}$) in SNR for large $N$, demodulation of the subcarrier via FFT yields \cite{OFDMRev-I}
\begin{equation}
  \label{eq:rec1}
   y_k = e^{j2\pi (k/N) n_\varepsilon} h_k  x_k + i_k +  w_k,
\end{equation}
where  $i_k$ is the disturbance caused by both ICI and ISI. The disturbance can be well approximated by Gaussian noise with power \cite{OFDMRev-I}
\begin{equation}
   \sigma_\varepsilon^2 \approx \sum_i |\alpha_i(t)|^2 \left(2 \frac{\Delta \varepsilon_i}{N} - \left(\frac{\Delta \varepsilon_i}{N}\right)^2 \right),
\end{equation}
where
\begin{eqnarray}
  \Delta \varepsilon_i = \left\{\begin{array}{c}  n_\varepsilon -\frac{\tau_i}{T_s}, \quad \quad \quad \quad \quad \quad  n_\varepsilon T_s > \tau_i \\
    \frac{\tau_i - T_g}{T_s} - n_\varepsilon, \quad  0< n_\varepsilon T_s <-(T_g - \tau_i) \\
   0,  \quad \quad \quad \quad \text{otherwise}\end{array}\right.
\end{eqnarray}

With a challenge-response process, ANA-PHY-PCRAS involves two rounds of communications. Hence,  the receiver imperfections from both Alice and Bob should be considered together. Let $n^a_\varepsilon$, $n^b_\varepsilon$ be the normalized sampling symbol offsets of Alice's and Bob's receivers, respectively. When Alice challenges, Bob estimates the channel phase at subcarrier $f_k$. With the sampling symbol offset $n^b_\varepsilon$, this phase estimate must include an extra increment over frequency, namely,
\begin{equation}
 \hat{\theta}_{k} =   \theta_{k} + 2\pi n^b_\varepsilon \cdot \frac{l_k}{N} + \theta^e_k,
\end{equation}
where $\theta^e_k$ is the non-biased estimation error with zero mean, and $l_k = l_0 + k \Delta \ell$.

When Bob responds to Alice, Alice also introduces her sampling symbol offset $n^a_\varepsilon$, and she can finally manage to obtain $L$ parallel channels at subcarriers $f_k, k=0, \cdots, L-1$ as
\begin{eqnarray}
\label{eq:fy}
   y_k =  \rho_k e^{j \theta_\varepsilon} e^{j(\varphi_k + k \varpi)} + i_k +  w_k,  k=0, \cdots, L-1
\end{eqnarray}
where $\varpi=  2\pi (n^a_\varepsilon - n^b_\varepsilon) \cdot \frac{\Delta \ell}{N}$, $\theta_\varepsilon = 2\pi (n^a_\varepsilon - n^b_\varepsilon) \cdot \frac{l_0}{N}$, $\rho_k = |h_k| e^{-j\theta^e_k}$ and $i_k$ denotes the interference due to the sampling offset $n^a_\varepsilon T_s$ at Alice.

\subsection{The Effect of Sampling Clock Frequency Offset }
With a sampling clock period of $T'_s$, the received samples at $t'_n = (n+ N_g)T'_s$ can be written as
\begin{eqnarray}
  r_n &\triangleq & r(t'_n) = \sum_i \alpha_i s(t'_n-\tau_i) = \sum_i \alpha_i \sum_k x_k e^{j2\pi \frac{k}{T_u} ((n+N_g)T'_s-T_g-\tau_i)} \nonumber \\
      &=& \sum_i \alpha_i \sum_k x_k e^{j2\pi \frac{k}{N} \left[n(1+\varsigma)+N_g \varsigma-\frac{\tau_i}{T_s}\right]} \nonumber \\
      &=& \sum_k \left(x_k e^{j2\pi k \frac{N_g\varsigma}{N}}\right) h_k e^{j2\pi n\frac{k+k\varsigma}{N}}.
\end{eqnarray}
Demodulation of the subcarrier yields \cite{OFDMRev-I}
\begin{equation}
  \label{eq:recSigErr}
   y_k = e^{j2\pi k (\frac{N_g\zeta}{N})} e^{j2\pi(\vartheta'\frac{N-1}{2N})} \frac{\sin(\pi \vartheta')}{N \sin(\frac{\pi \vartheta'}{N})}  h_k  x_k + i_k + w_k,
\end{equation}
where $\vartheta' = k\varsigma$ and $i_k$ is the disturbance caused by ICI.

Consider a sampling clock frequency offset up to $\pm$100 ppm ($\varsigma=10^{-4}$) for an OFDM system of $N=2048$ subcarriers.  The multiplicative factor $\frac{\sin(\pi \vartheta')}{N \sin(\frac{\pi \vartheta'}{N})}$ results in some loss in SNR, which is less than 0.3 dB in the worst carrier. The sampling frequency offset also results in an incremental phase rotation over subcarriers, which is the same to (\ref{eq:fy}).

\subsection{Verification under Practical Imperfections}
With a challenge-response approach, we focus on the final verification in the response stage. As depicted in Section-III.C, an equivalent frequency offset due to sampling offset at the stage of challenging should be considered.

By including all the above imperfections, the demodulated subcarrier at $f_k$ is given by
\begin{equation}
  \label{eq:imodel}
   y_k =  e^{j( k \varpi + \phi_0)}  h_k  x_k + i_k + w_k,
\end{equation}
where
\begin{eqnarray}
\label{eq:resfq}
 \varpi = 2\pi \frac{(n^a_\varepsilon - n^b_\varepsilon) \Delta \ell + N_g\varsigma + (N-1)\varsigma/2}{N},
\end{eqnarray}
\begin{eqnarray}
\phi_0 = \pi \Delta f T_u (N-1)/N + 2\pi (n^a_\varepsilon - n^b_\varepsilon) l_0/N +  \Phi_0,
\end{eqnarray}
and $i_k$ is the disturbance caused by both ICI and ISI.

Consider the modified ANA-PHY-PCRAS for the time-separated subchannels. With the channel model (\ref{eq:imodel}) under practical imperfections, we propose to employ a refined non-coherent combining metric
\begin{eqnarray}
  \label{eq:rmod}
   \zeta = \max_{\varpi}\sum_{m=1}^J \left|\mathcal{K}_m^\dag \Lambda(\varpi) \mathbf{y}(t_m)\right|^2,
\end{eqnarray}
where $\Lambda(\varpi)=\text{diag}(1, e^{-j\varpi}, e^{-j 2\varpi},\cdots, e^{-j (L-1)\varpi})$, and $J$ time slots starting at $t_m, m=0, 1, \cdots, J-1$ are employed. Compared to (\ref{eq:mod1}), the refined metric includes the effect of residual frequency-offset (\ref{eq:resfq}) due to various imperfections.

For the Scenario. \ref{s1} with $\Delta \ell=128$, we have that $\frac{\Delta \ell}{N}=\frac{1}{16}$, which can result in a very large frequency offset (\ref{eq:resfq}) even with a small value of $|n^a_\varepsilon - n^b_\varepsilon|$. Therefore, the search of frequency shown in (\ref{eq:rmod}) should be seriously considered in practice. Noting that the contribution of $\frac{N_g\varsigma + (N-1)\varsigma/2}{N}$ in (\ref{eq:resfq}) due to sampling clock frequency offset is minor compared to sampling offset.

\section{Security Analysis}
In this section, security analysis is presented. For ease of analysis, we focus on the basic ANA-PHY-PCRAS over a single OFDM symbol.

\subsection{Noncoherent Channel Model for Eavesdropping}
As a passive attacker, Eve only monitors all frames inside the network during authentication, and tries to learn $(\mathcal{K}_A, \mathcal{K}_B)$ from whatever it gets.

By monitoring the response signal from Bob, the received signal at Eve is given by
\setlength{\arraycolsep}{0.0em}
\begin{eqnarray}
  \label{eq:pa}
     r_E(t) =\sum_{k=0}^{L-1}  |\tilde{h}_k| \cos \left(2\pi f_k t + (\varphi_k - \hat{\theta}_k + \upsilon_k) + \tilde{\theta}_k \right) + w_E(t),
\end{eqnarray}
\setlength{\arraycolsep}{0.5em}where $\tilde{h}_k = |\tilde{h}_k|e^{j\tilde{\theta}_k}$, $\tilde{\theta}_k$ is Eve's channel-phase response when Bob transmits a zero-phase sinusoidal signal at frequency $f_k$, $\hat{\theta}_k$ is Bob's estimate of channel response $\theta_k$ when Alice challenges,  and $w_E(t)$ is the noise process observed by Eve.

Due to the orthogonality among different subcarriers, one can retrieve the discrete signal vector from (\ref{eq:pa}) as $z_0^{L-1}=[z_0,\cdots,z_{L-1}]^T$, where
\begin{eqnarray}
\label{eqn:p1}
   z_k = |\tilde{h}_{k}|e^{j\psi_k} e^{j \varphi_{k}} + w_{k},
\end{eqnarray}
and $\psi_k = (\tilde{\theta}_{k} - \hat{\theta}_{k}) + \upsilon_k$.

For security analysis, we focus on the key equivocation or the conditional equivocation about the key, namely, $H(\mathcal{K}_B | z_0^{L-1})$.
As
\begin{eqnarray}
   H(\mathcal{K}_B | z_0^{L-1}) = H(\mathcal{K}_B) - I(z_0^{L-1}; \mathcal{K}_B),
\end{eqnarray}
where $I(X;Y)$ denotes the mutual information between two random variables $X$ and $Y$, it is equivalent to compute the mutual information $I(z_0^{L-1}; \mathcal{K}_B)$ or its bound.
If $I(z_0^{L-1}; \mathcal{K}_B) \le \delta H(\mathcal{K}_B)$, it follows that $H(\mathcal{K}_B | z_0^{L-1}) \ge (1-\delta) H(\mathcal{K}_B)$. Hence,  the successful probability for an eavesdropper to guess the key is about $2^{-(1-\delta) |\mathcal{K}_B|}$. In the ideal case of $I(z_0^{L-1}; \mathcal{K}_B)=0$, we have that $H(\mathcal{K}_B | z_0^{L-1})= H(\mathcal{K}_B)$, which means that the successful probability for an eavesdropper to guess the key is about $2^{-|\mathcal{K}_B|}$, the same as a random guess. Whenever $I(z_0^{L-1}; \mathcal{K}_B) = 0$, information-theoretic security is ensured.

With a noncoherent metric for verification, the shared key $\mathcal{K}_B$ is essentially conveyed in the differences of modulated phase sequence $\varphi_0^{L-1}$. This means that we are interested in the noncoherent channel model of (\ref{eqn:p1}), where the mutual information $I(z_0^{L-1}; \mathcal{K}_B)$ is determined by the sequence of phase differences $\left\{\Delta \psi_k =\psi_k - \psi_{k-1}\right\}_{k=1}^L$, but not on $\psi_0$. To be more rigourous for security analysis, we always assume that Eve has the complete knowledge about the channel, which means that $\Delta \tilde{\theta}_k =0$ (as it can be perfectly compensated by Eve). Since $\Delta \psi_k = \Delta \tilde{\theta}_k - \Delta \hat{\theta}_k + \Delta \upsilon_k$, we have that $\Delta \psi_k = -\Delta \hat{\theta}_k + \Delta \upsilon_k$, or
\begin{equation}
\label{eq:phi}
\psi_k = - \hat{\theta}_k + \upsilon_k + \lambda,
\end{equation}
where $\lambda$ denotes an unknown but constant phase rotation over the subchannel index $k$. Here, $\lambda$ is often assumed to be uniformly distributed over $(-\pi,\pi]$.

\subsection{Information-Theoretic Security under Independent Parallel Fading Channels}
For wireless rich-scattering fading channels, the observations of Eve remain independent from the channel-specific observations of
Alice and Bob, if Eve is located more than half a wavelength away from these two users \cite{SecureTrans,ITYE}. In this case, Eve cannot get a feasible estimate about $\theta_{k}$ based on the monitoring signal when Alice initiates a challenge. Hence, it is fair to assume that Eve has no any knowledge about either $\theta_k$ or $\hat{\theta}_k$.

\begin{lem}
Let $\theta_1, \theta_2 \in (-\pi, \pi]$ be two random variables on a circle and $\theta=\theta_1+\theta_2 \mod 2\pi$, where $\theta \in (-\pi, \pi]$. If $\theta_1$ is uniformly distributed over $(-\pi, \pi]$ and $\theta_2$ is independent of $\theta_1$,  it follows that $\theta$ is also uniformly distributed over $(-\pi, \pi]$, which is irrespective of the distribution of $\theta_2$.
\end{lem}
\begin{proof}
Let $f_{\theta_1} (x), f_{\theta_2}(x), f_\theta(x)$ denote the pdfs of $\theta_1, \theta_2, \theta$, respectively. For a uniformly distributed random variable on a circle, we have that $f_{\theta_1}(x)=\frac{1}{2\pi}$ if $x\in (-\pi,\pi]$, zeros otherwise. Since $\theta_2$ is independent of $\theta_1$, it follows that
\begin{eqnarray*}
  f_\theta(x) = \int_{-\pi}^\pi f_{\theta_1}(t) f_{\theta_2}(x-t)dt = \frac{1}{2\pi}\int_{-\pi}^\pi f_{\theta_2}(x-t)dt=\frac{1}{2\pi}
\end{eqnarray*}
for $x\in (-\pi, \pi]$.
\end{proof}

If the $L$ parallel fading channels at subcarriers $f_k, k=0, 1, \cdots, L-1$ between Alice and Bob are independent, we have that either $\theta_{k}$  or their estimates $\hat{\theta}_{k}, k=0, \cdots, L-1$ are i.i.d, each of which is uniformly distributed over $(-\pi,\pi]$. Since Eve's channel phase response $\tilde{\theta}_k$ is independent of $\hat{\theta_k}$ and by noting Lemma 1,  it is clear that $\psi_k, k=0, 1, \cdots, L-1$ (\ref{eq:phi}) are also i.i.d and uniformly distributed over $(-\pi,\pi]$. This means that
\begin{eqnarray}
 I(z_0^{L-1}; \mathcal{K}_B) = 0.
\end{eqnarray}
Therefore, there is no hope for Eve to extract any reliable information about the key $\mathcal{K}_A$. In this case, information-theoretic security can be perfectly ensured.

\subsection{Equivocation Analysis for Static Parallel Channels}
The worst case for the purpose of authentication is to consider the scenario, where the $L$ parallel channels between Bob and Alice (or Eve) are all assumed to be static over a long period. This means that  $\theta_k$ can be well estimated before the start of authentication and further compensated in (\ref{eq:phi}) by Eve, who may get a clean version of the received signal
\begin{eqnarray}
\label{eqn:EveCl}
   z_k = |\tilde{h}_{k}| e^{j(\varphi_k + \upsilon_k + \lambda)} + w_k, k=0,1,\cdots,L-1.
\end{eqnarray}

As Eve can be located very close to Bob, her observation may be free of noise, which is the worst case for addressing the security issue.  In this case, Eve can directly extract the phase of $z_k$, namely,
\begin{eqnarray}
    \phi_k = \varphi_k + \upsilon_k + \lambda, k=0,1, \cdots, L-1
\end{eqnarray}
where $\phi_k=\angle(z_k)$.

Hence, the mutual information between $z_{0}^{L-1}$ and $\mathcal{K}_B$ can now be computed as
\begin{eqnarray}
   I(z_0^{L-1}; \mathcal{K}_B) = I(\phi_0^{L-1}; \varphi_0^{L-1}) = E_{\phi_0^{L-1},\varphi_0^{L-1}}\log_2 \frac{p(\phi_0^{L-1}|\varphi_0^{L-1})}{p(\phi_0^{L-1})},
\end{eqnarray}
where
\begin{eqnarray}
  p(\phi_0^{L-1}|\varphi_0^{L-1}) &=&  \int_\lambda p(\phi_0^{L-1}|\varphi_0^{L-1}, \lambda)p(\lambda) d\lambda \nonumber \\
                                  &=&  \int_\lambda \prod_k f_\upsilon(\phi_k-\varphi_k-\lambda)p(\lambda) d\lambda \nonumber \\
                                  &=& \frac{I_1\left(\beta\sqrt{\left(\sum_{k=0}^{L-1}\cos (\phi_k- \varphi_k) \right)^2 + \left(\sum_{k=0}^{L-1}\sin (\phi_k- \varphi_k) \right)^2}\right)}{\left[2\pi I_0(\beta)\right]^L}
\end{eqnarray}
with $I_1(x)$ denoting the modified Bessel function of the first kind and 1-th order.  Since $\varphi_k$'s are uniformly, i.i.d. over the discrete values,
\begin{eqnarray}
  p(\phi_0^{L-1}) = \sum_{\varphi_0^{L-1} \in \Omega^L} p(\phi_0^{L-1}|\varphi_0^{L-1}) P(\varphi_0^{L-1}),
\end{eqnarray}
where $P(\varphi_0^{L-1})=\frac{1}{M^L}$ for equally-distributed $M$-PSK constellations.

\begin{lem}
Consider the worst case scenario, where the channel between Alice and Bob is static and Eve can get a noise-free version of the transmitted signal by either Alice or Bob. With ANA-PHY-PCRAS, Eve's key equivocation can be lower bounded as
\begin{eqnarray}
\label{eq:lowBound}
   H(\mathcal{K}_B | Z_0^{L-1}) \ge  L\cdot E_{\varphi,\upsilon}\left[\log_2\frac{\sum_{\bar{\varphi}\in \Omega}f_\upsilon(\varphi+\upsilon-\bar{\varphi})}{f_\upsilon(\upsilon)}\right]
\end{eqnarray}
if the introduced artificial noise is with the pdf of $f_\upsilon(x)$.
\end{lem}

\begin{proof}
It was shown in \cite{PelegCOM} that
\begin{eqnarray}
  I(\phi_0^{L-1}; \varphi_0^{L-1}) &=& I(\phi_0^{L-1}; \varphi_0^{L-1}|\lambda) - \left[I(\lambda;\phi_0^{L-1}|\varphi_0^{L-1}) - I(\lambda;\phi_0^{L-1}) \right] \nonumber \\
   &=& I(\phi_0^{L-1}; \varphi_0^{L-1}|\lambda) - I(\lambda;\phi_0^{L-1}|\varphi_0^{L-1}),
\end{eqnarray}
where  $I(\lambda;\phi_0^{L-1})=0$ as $\{\lambda+\varphi_k\}_{k=0}^{L-1}$ is independent of $\lambda$, and the first term $I(\phi_0^{L-1}; \varphi_0^{L-1}|\lambda)$ denotes the coherent mutual information. By assuming a coherent channel model of $\phi =\varphi + \upsilon$, it can be efficiently computed as
\begin{eqnarray}
\label{eqn:muinfc}
\frac{1}{L} I(\phi_0^{L-1}; \varphi_0^{L-1}|\lambda) &=&I(\phi; \varphi) \nonumber \\
                &=&   E_{\phi,\varphi}\log_2\frac{p(\phi|\varphi)}{p(\phi)} \nonumber \\
                &=& E_{\phi,\varphi}\log_2\frac{p(\phi|\varphi)}{\frac{1}{M}\sum_{\bar{\varphi}\in \Omega}p(\phi|\bar{\varphi})} \nonumber \\
                &=&  \log_2 M -  E_{\phi,\varphi}\left[\log_2\frac{\sum_{\bar{\varphi}\in \Omega}p(\phi|\bar{\varphi})}{p(\phi|\varphi)}\right] \nonumber \\
                &=& \log_2 M -  E_{\varphi,\upsilon}\left[\log_2\frac{\sum_{\bar{\varphi} \in \Omega}f_\upsilon(\varphi+\upsilon-\bar{\varphi})}{f_\upsilon(\upsilon)}\right].
\end{eqnarray}
By noting that  $I(\phi_0^{L-1}; \varphi_0^{L-1}) \le I(\phi_0^{L-1}; \varphi_0^{L-1}|\lambda)$, the conditional equivocation can be bounded as
\begin{eqnarray}
   H(\mathcal{K}_B | Z_0^{L-1}) &=& H(\mathcal{K}_B) - I(Z_0^{L-1}; \mathcal{K}_B) \nonumber \\
   &\ge&  L\cdot E_{\varphi,\upsilon}\left[\log_2\frac{\sum_{\bar{\varphi}\in \Omega}f_\upsilon(\varphi+\upsilon-\bar{\varphi})}{f_\upsilon(\upsilon)}\right],
\end{eqnarray}
which could be strictly positive for a properly chosen distribution $f_\upsilon(x)$.
\end{proof}

\section{Numerical Examples}
\subsection{An Application Model for Getting the Shared Keys}
%\begin{figure*} %[htbp]
\begin{figure}[htb]
   \centering
   \includegraphics[width=0.50\textwidth]{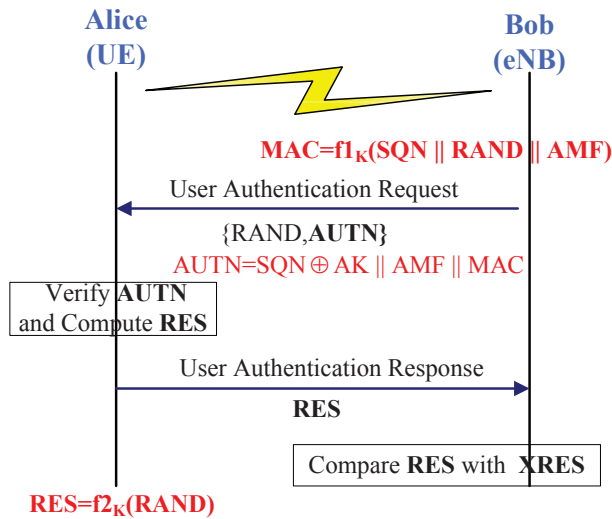}
   \caption{A Typical Challenge-Response Authentication Process.}
   \label{fig:sysModel}
\end{figure}

In developing ANA-PHY-PCRAS, we have assumed that Alice and Bob share two secret keys, namely, $\{\mathcal{K}_A,\mathcal{K}_B\}$. In practical wireless network scenarios, it is interesting to investigate how Alice and Bob can share secrets before authentication. This, indeed, depends on the underlying wireless network.

For 4G mobile networks, we consider a typical scenario where a user equipment (UE/Alice), wants to authenticate with an evolved Node-B (eNB/Bob). The possibility of sharing common secrets between Alice and Bob comes from the long-term secret key ($K$) stored on the Universal Subscriber Identity Module (USIM) and in the Authentication Center (AuC). The challenge-response authentication process can be depicted in Fig. \ref{fig:sysModel}. As shown, a pair of shared keys $\{\mathcal{K}_A,\mathcal{K}_B\}$ can be derived from the long-term key $K$, namely,
\begin{eqnarray}
\mathcal{K}_A &=& f2_K(\text{RAND}), \nonumber \\
\mathcal{K}_B &=& f1_K(\text{SQN}||\text{RAND}||\text{AMF}),
\end{eqnarray}
where RAND, SQN, AMF can be considered as random numbers, and $f1, f2$ are the message authentication function used to compute MAC and RES (XRES), respectively. Please refer to \cite{3GPPSA} for more details.

\subsection{Simulation Scenario}
Consider that the system operates at carrier frequency of $1.9$ GHz with a bandwidth of $W=20$ MHz, which is divided into $N=2048$ tones with a total symbol period of 108.8 $\mu$s, of which 6.4 $\mu$s constitutes the CP. Hence, $N_g=128$ and $N_f=N+N_g=2176$. $N=2048$ parallel subchannels are obtained using both IFFT and FFT. For ANA-PHY-PCRAS, $L=64+1$ subchannels with equal bandwidth interval ($\Delta \ell =32$) are selected with the minimum normalized correlation of 0.7136 among $L$ subchannels. For the modified ANA-PHY-PCRAS, $L'=16+1$ subchannels with equal bandwidth interval ($\Delta \ell =128$) are selected with the minimum normalized correlation of 0.2468. To allocate $L=L' J$ sub-channels, we repeatedly employ such $L'$ carriers at times $t_0, t_1, \cdots, t_{J-1}$, where $t_j=t_0 + j \cdot \delta T$ and $\delta T = 10 T_c =48$ms can be employed for example. With the use of large $\delta T$, the allocated subchannels at different time slots could be nearly uncorrelated. However, it should be noted that the use of large $\delta T$ could cause noticeable end-to-end delay. In the case of $\delta T = 10 T_c =48$ms and $J=4$, the end-to-end delay is at least $J \cdot \delta T =192$ ms, which is comparable to the time delay due to authentication in LTE \cite{3GPPSA}.

In simulations, we employ the channel model with exponentially decaying power-delay profile, where a total of 20 multipaths are assumed, the normalized delays $\dot{\tau}_i, i=0,1,\cdots, 19$ are assumed to be uniformly and independently distributed over the length of CP ($\dot{\tau}_i \in [0,N_g]$), and $\sigma_\tau = 0.5 \mu$s. This channel model is comparable to the urban channel defined in \cite{3GPPCh}, with 20 multipaths and maximum delay spread of $2.14 \mu$s. The path gains $\alpha_i (t)$'s are assumed to be complex-Gaussian distributed, which remain constant during one OFDM symbol but varying independently if the time interval between two OFDM symbols is larger than $\delta T$.
\begin{figure}[htb] %[htbp]
   \centering
   \includegraphics[width=0.5\textwidth]{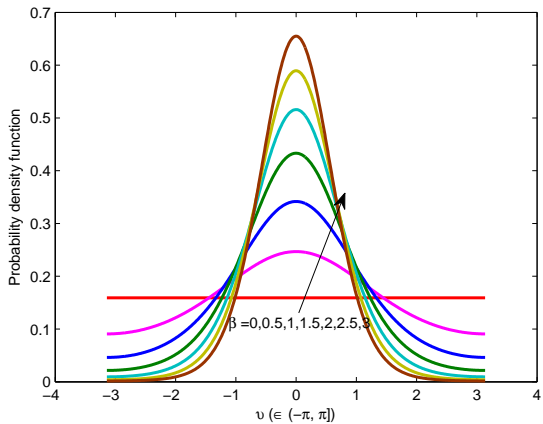} %{PWBF_BER_Uth.eps}
   \caption{Probability density functions of Tikhonov distributed artificial noise with different $\beta$'s.}
   \label{fig:TikPdf}
\end{figure}

\begin{figure}[htb] %[htbp]
   \centering
   \includegraphics[width=0.5\textwidth]{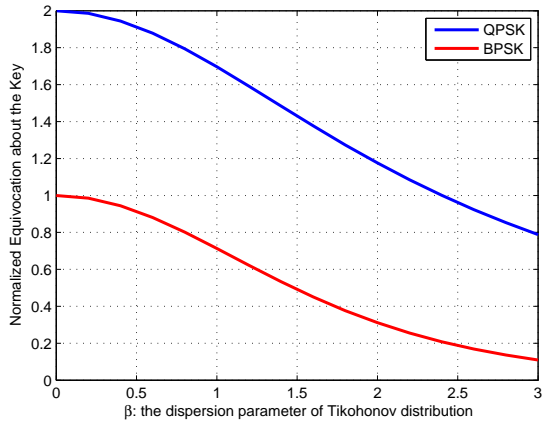} %{PWBF_BER_Uth.eps}
   \caption{Normalized equivocation about the key with ANA-PHY-PCRAS.}
   \label{fig:Eqv}
\end{figure}

For the design of physical layer authentication schemes, one should carefully balance the three performance metrics, namely, the successful authentication rate, the false acceptance rate and the (normalized) key equivocation $\frac{1}{L}H(\mathcal{K}_B | Z_0^{L-1})$ for any eavesdropper.  \textit{In most scenarios, the ideal Receiver Operating Characteristic (ROC) (successful authentication rate versus false acceptance rate) can be achieved without much difficulty in the working SNR region for the purpose of communications. Hence, the key equivocation, as a security metric, could be of the first importance for its use in practice}.

\subsection{Key Equivocation}
 We compute the key equivocation for ANA-PHY-PCRAS in the worst case scenario. As shown in (\ref{eq:lowBound}), it depends on the specified distribution of artificial noise. Fig. \ref{fig:TikPdf} shows the pdfs of the Tikhonov distributed artificial noise with different $\beta$'s. Then, we plot the lower bound (\ref{eq:lowBound}) on the (normalized) key equivocation in Fig. \ref{fig:Eqv} for different $\beta$'s, with both BPSK and QPSK constellations.
As shown, the key equivocation achieves the maximum at $\beta=0$, in which case the uniformly-distributed artificial noise over $(-\pi,\pi]$ is employed. The key equivocation decreases when $\beta$ increases. When no artificial noise is introduced, the key equivocation is simply reduced to zero for this worst case scenario, which means that there is no guarantee of information-theoretic security for PHY-PCRAS \cite{WuCL2014}. Clearly, the use of higher-order modulation scheme can strengthen the system security as the key equivocation increases.

\textit{We comment here that there is simply no guarantee of information-theoretic security for various reported physical-layer authentication schemes \cite{DanPHY,YuIFS} if Eve is very close to Bob and hence she can get a noise-free version of the transmitted signal by Bob, and in the same time the channel between Alice and Bob is unfortunately static over the period of authentication.}

\subsection{ROC Performance}
Through extensive Monte-Carlo simulations, we investigate the pdfs of $\zeta$ under two hypothesis $H_i, i=0,1$, which can be well employed to evaluate both successful authentication and false acceptance rates.  The proper choice of the threshold $\iota$ can also be determined from the pdfs of $\zeta$.

\subsubsection{PHY-PCRAS, ANA-PHY-PCRAS and Modified ANA-PHY-PCRAS}
\begin{figure}[htb] %[htbp]
   \centering
   \includegraphics[width=0.5\textwidth]{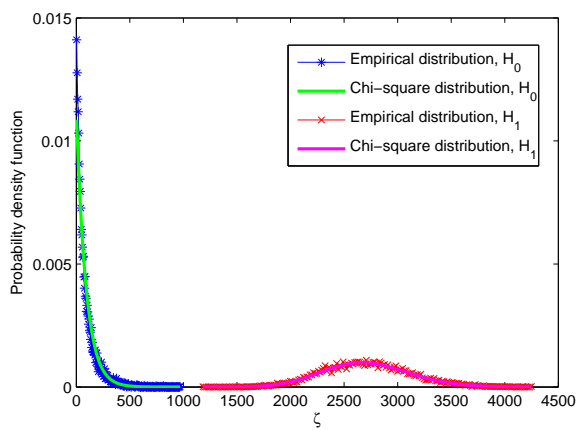} %{PWBF_BER_Uth.eps}
   \caption{Probability density functions of $\zeta|H_1$ and $\zeta|H_0$ at SNR = 5 dB for PHY-PCRAS.}
   \label{fig:pdf1}
\end{figure}

\begin{figure}[htb] %[htbp]
   \centering
   \includegraphics[width=0.5\textwidth]{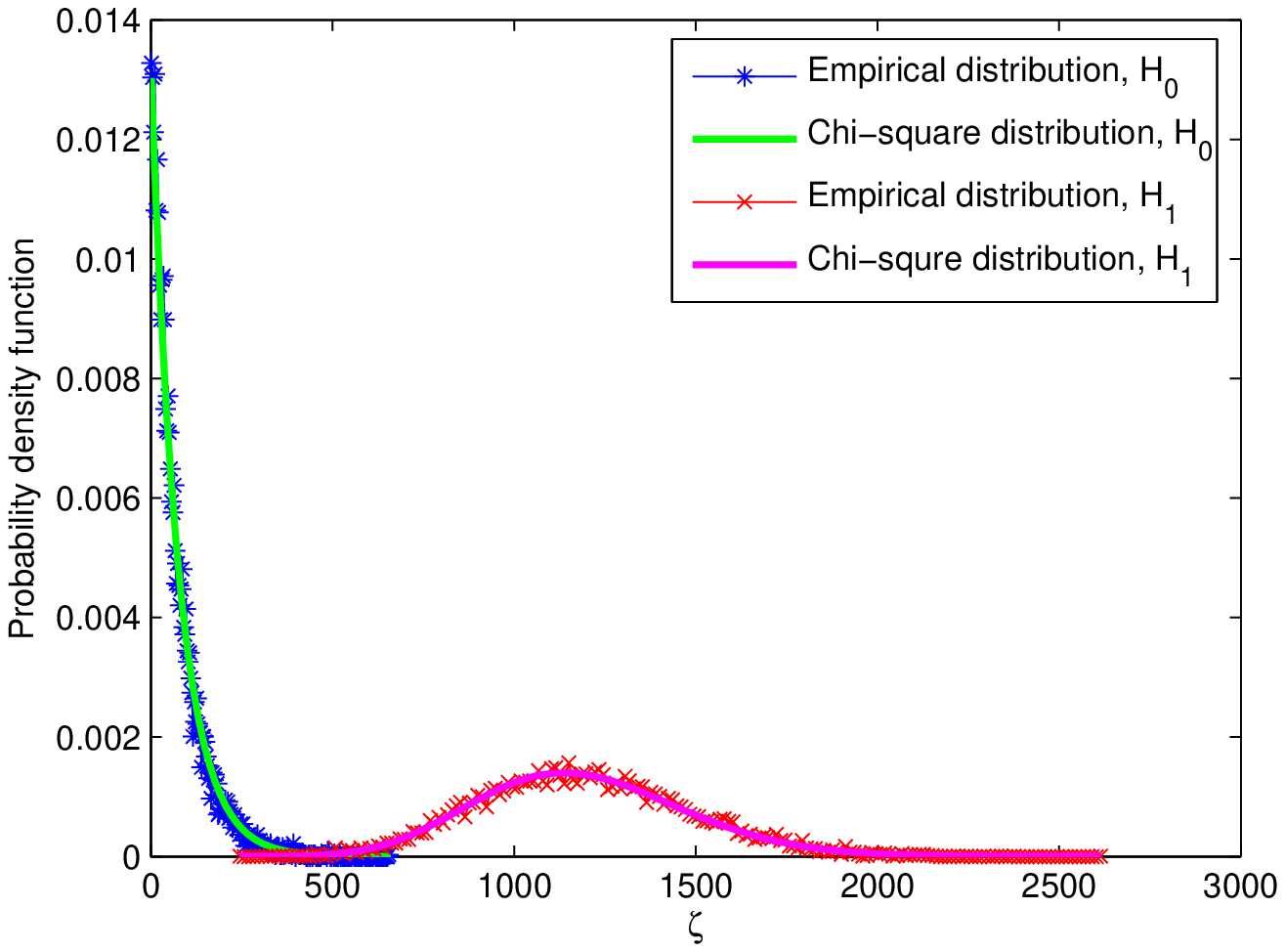} %{PWBF_BER_Uth.eps}
   \caption{Probability density functions of $\zeta|H_1$ and $\zeta|H_0$ at SNR = 10 dB and $\beta=1.5$ for ANA-PHY-PCRAS.}
   \label{fig:pdf2}
\end{figure}

\begin{figure}[htb] %[htbp]
   \centering
   \includegraphics[width=0.5\textwidth]{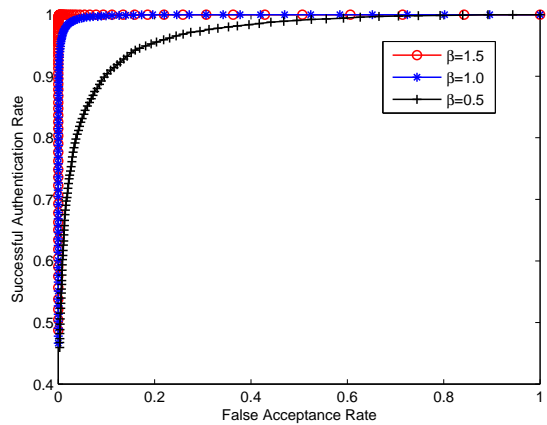} %{PWBF_BER_Uth.eps}
   \caption{Successful authentication rate versus false acceptance rate at SNR=10 dB for ANA-PHY-PCRAS with different $\beta$'s.}
   \label{fig:rocbeta}
\end{figure}

\begin{figure}[htb] %[htbp]
   \centering
   \includegraphics[width=0.5\textwidth]{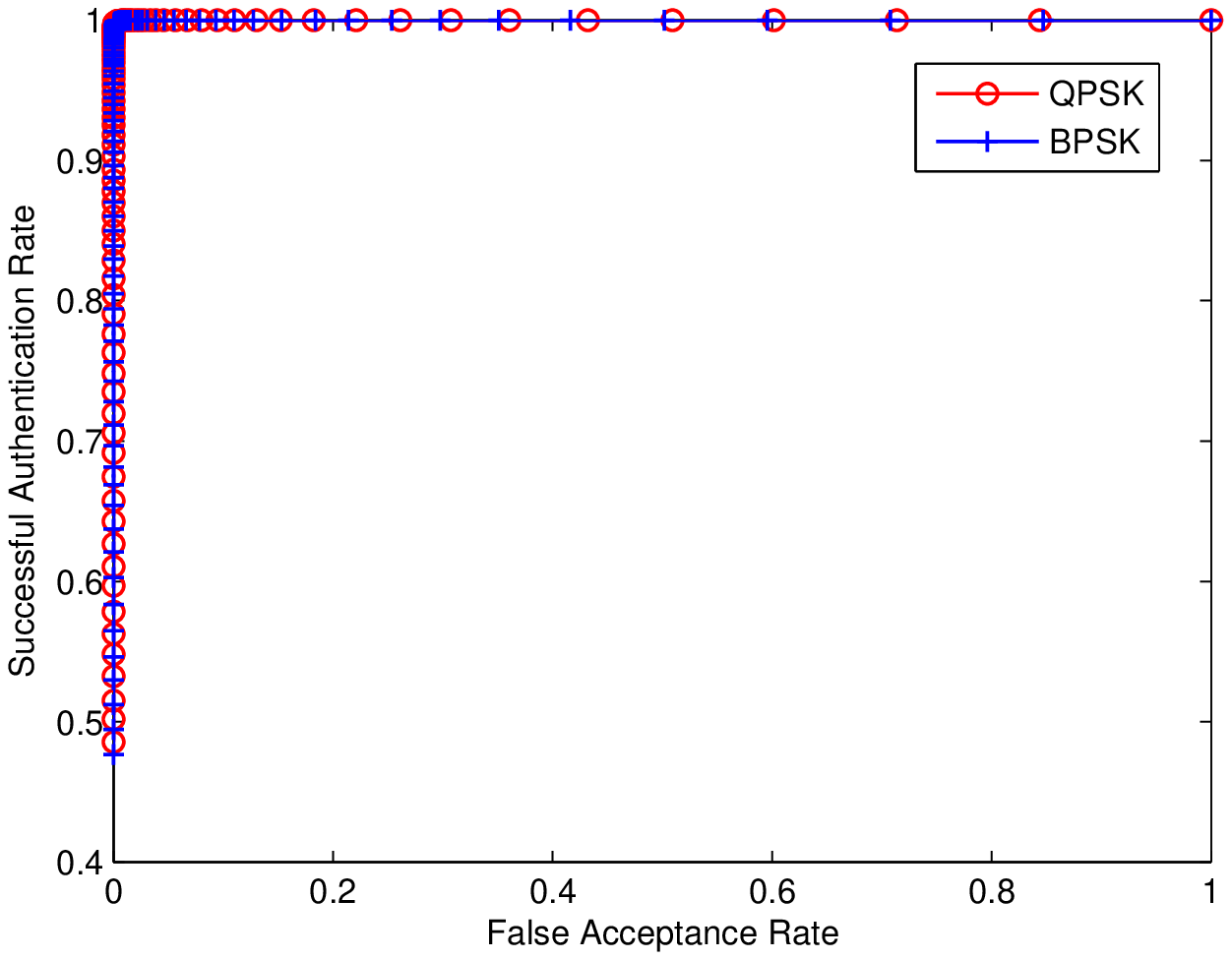} %{PWBF_BER_Uth.eps}
   \caption{Comparison of ROC curves with ANA-PHY-PCRAS for both BPSK and QPSK constellations ($\beta=1.5$).}
   \label{fig:rocBQ}
\end{figure}

With $L=64+1$ subchannels selected among $N=2048$ OFDM subchannels ($\Delta \ell=32$) for a single OFDM symbol, Fig. \ref{fig:pdf1} shows empirical pdfs of $\zeta|H_1$ and $\zeta|H_0$ at SNR=5 dB for PHY-PCRAS, while Fig. \ref{fig:pdf2} shows empirical pdfs of $\zeta|H_1$ and $\zeta|H_0$ at SNR=10 dB for ANA-PHY-PCRAS with $\beta=1.5$. In both figures, BPSK constellation is assumed. As claimed in Section-III, $\zeta|H_1$ and $\zeta|H_0$ are both Chi-square distributed. Hence, Chi-square distributions are also given in both figures, where $|\bar{\eta}_i|, \sigma^2_{H_i}, i=0,1$ are directly estimated through Monte-Carlo simulations \cite{RiceEst}. As shown, the theoretical Chi-square distributions are coincided well with the empirical distributions even though $L$ subchannels are correlated.  Since the pdf of $\zeta|H_1$ is far apart from that of $\zeta|H_0$ even at the SNR of 5 dB in Fig. \ref{fig:pdf1}, almost ideal ROC curve can be observed. With the introduction of artificial noise, the ROC performance of ANA-PHY-PCRAS is clearly inferior to that of PHY-PCRAS as indicated by Fig. \ref{fig:pdf2}.

Next, we investigate the effect of $\beta$ on the ROC curves for ANA-PHY-PCRA, which is depicted in Fig. \ref{fig:rocbeta} for different $\beta$'s. From both Fig. \ref{fig:rocbeta} and Fig. \ref{fig:Eqv}, we conclude that there is a fundamental tradeoff between the ROC performance and security, which is controlled by the amount of artificial noise ($\beta$).

In Fig. \ref{fig:Eqv}, we have shown that the use of QPSK constellation is significantly superior to the use of BPSK constellation for the security of ANA-PHY-PCRAS. Here, we show their ROC curves in Fig. \ref{fig:rocBQ} for ANA-PHY-PCRAS with both BPSK and QPSK constellations, where $\beta=1.5$ is used. Noting that the use of QPSK constellation requires the size of key doubled compared to the use of BPSK constellation. As shown, the same ROC curves are observed for both BPSK and QPSK. Hence, the use of higher order constellations can significantly improve the security of ANA-PHY-PCRAS, which is very helpful in practical scenarioes whenever the number of allocated subchannels is not enough compared to the size of the shared key.

\begin{figure}[htb] %[htbp]
   \centering
   \includegraphics[width=0.5\textwidth]{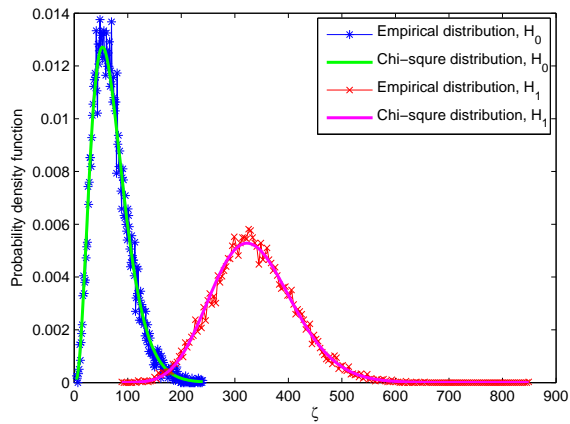} %{PWBF_BER_Uth.eps}
   \caption{Probability density functions of $\zeta|H_1$ and $\zeta|H_0$ at SNR=10 dB with time-separated subchannel allocation ($J=4$) and $\beta=1.5$.}
   \label{fig:pdf3}
\end{figure}

We also consider the modified ANA-PHY-PCRAS, where $L=L'J$ subchannels allocated for four ($J=4$) far-separated OFDM symbols with $L'=16+1$ subchannels allocated for each OFDM symbol. Fig. \ref{fig:pdf3} shows its empirical pdfs of $\zeta|H_1$ and $\zeta|H_0$ at SNR=10 dB and $\beta=1.5$. Although the modified ANA-PHY-PCRAS can be better protected by the randomness of the physical channel (due to well-separated subchannels in both time and frequency), it, however, is slightly inferior to ANA-PHY-PCRAS in the ROC performance as indicated in Fig. \ref{fig:pdf2} and Fig. \ref{fig:pdf3}, due to noncoherent combining loss.

\subsubsection{The effect of practical imperfections}
We consider practical imperfections in both the challenge and response stages. Imperfects at the receiver of Bob in the challenge stage are assume to be independent from the receiver of Alice in the response stage.

In simulations, both the effects of carrier frequency offset and sampling offset are considered, while the sampling frequency offset is not considered, as its effect can be well included in the equivalent channel model as shown in (\ref{eq:imodel}).  The residual carrier frequency offset $\vartheta=\Delta f T_u$ is assumed to be uniformly distributed in $[-\vartheta_{\max},\vartheta_{\max}]$. The sampling offset $n_\varepsilon$ is also uniformly distributed in $[-n_\varepsilon^{\max},n_\varepsilon^{\max}]$. By referring to (\ref{eq:rmod}), the verification should be searched over the range of $\varpi$, due to the sampling offsets introduced by the receivers at the sides of both Bob and Alice.
Clearly,
\begin{equation}
\varpi \in [-\varpi_{\max}, \varpi_{\max}], \varpi_{\max}=2\pi \times \frac{2n_\varepsilon^{\max}\Delta \ell}{N}
\end{equation}
With a step size of $2\pi \frac{2\varpi_{\max}}{N_w}$ for search of $\varpi$, there are $N_w$ candidate frequencies to be tested for maximizing $\zeta$ (\ref{eq:rmod}).

In Fig. \ref{fig:pdf4}, the modified ANA-PHY-PCRAS is considered for $J=4, n_\varepsilon^{\max} = 10, \Delta \ell=128$, $\vartheta_{\max}=0.1$ and $\beta=1.5$. Clearly, $\varpi\in 2\pi\times [-0.625,0.625]$. One can show that the SNR loss \cite{OFDMRev-I} due to both carrier frequency offset and sampling offset is negligible when the working SNR is $10$ dB, which was verified by extensive simulations.

By comparing Fig. \ref{fig:pdf4} with Fig. \ref{fig:pdf3}, there is actually minor difference between the scenarios of zero- and non-zero sampling/carrier frequency offsets for the empirical pdfs when $N_w$ is set to 200. Even with $N_w=40$, it still works with slightly degraded ROC performance. Therefore, the number of candidate frequencies to be tested can be very small for authentication, and the increase in complexity due to the search of frequency can be well controlled.
\begin{figure}[htb] %[htbp]
   \centering
   \includegraphics[width=0.5\textwidth]{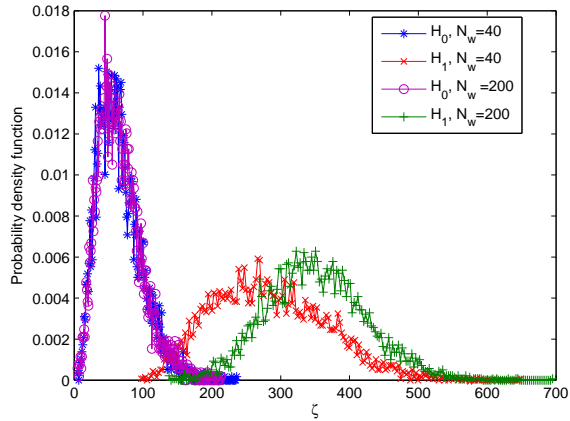} %{PWBF_BER_Uth.eps}
   \caption{Empirical probability density functions of $\zeta|H_1$ and $\zeta|H_0$ at SNR=10dB ($J=4$) and $\beta=1.5$.}
   \label{fig:pdf4}
\end{figure}

\subsection{Comparison with PHY-CRAM}
\begin{figure}[htb] %[htbp]
   \centering
   \includegraphics[width=0.5\textwidth]{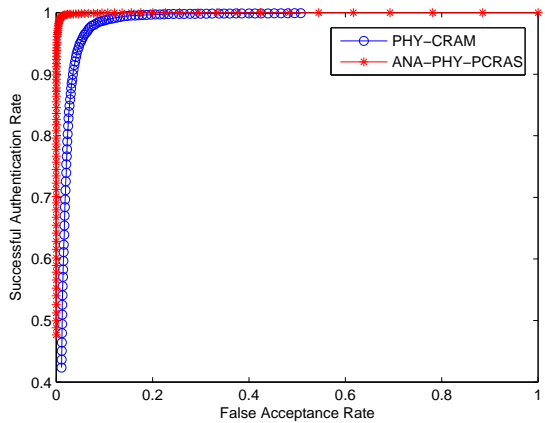} %{PWBF_BER_Uth.eps}
   \caption{Comparison of the proposed ANA-PHY-PCRAS and PHY-CRAM for ROC curves at SNR=5 dB and $L=64$.}
   \label{fig:comp}
\end{figure}
As a mutual physical challenge-response authentication scheme, the PHY-CRAM proposed in \cite{DanPHY} was shown to be simple, low complexity,
robust, and flexible. Hence, it is interesting to compare ANA-PHY-PCRAS with PHY-CRAM.

Fig. \ref{fig:comp} shows the comparison result in the ROC performance at SNR=5 dB, where $\beta=1.5$ is used for ANA-PHY-PCRAS. Therefore, a normalized key equivocation of $\frac{1}{L} H(\mathcal{K}_B | Z_0^{L-1}) \ge 0.491$ can be achieved in the worst case scenario. This, however, is not true for PHY-CRAM. Even with the introduction of artificial noise, ANA-PHY-PCRAS is still better than PHY-CRAM in the ROC performance as shown in Fig. \ref{fig:comp}. Indeed, PHY-CRAM employs amplitude modulation, which is often worse than phase modulation in performance. For implementation, high peak fluctuations may occur with PHY-CRAM, due to the employment of amplitude modulation. Hence, it requires to suppress the high peak in practice with additional complexity. ANA-PHY-PCRAS, however, is more sensitive to the frequency offset compared to PHY-CRAM.

\section{Conclusion}
In this paper, we proposed a novel ANA-PHY-PCRAS for practical OFDM transmission, where the Tikhonov-distributed artificial noise is introduced to interfere with the phase-modulated key for resisting potential key-recovery attacks. Thanks to the introduced artificial noise, the proposed ANA-PHY-PCRAS was proved to be secure even in the worst case scenario, where a static channel between Alice and Bob is assumed, and Eve can even get a noise-free version of the transmitted signal by either Alice or Bob.

Various practical issues are addressed for ANA-PHY-PCRAS with OFDM transmission, including correlation among subchannels, imperfect carrier and timing recoveries. The effect of sampling offset was shown to be significant for the practical implementation of ANA-PHY-PCRAS, and a search procedure in the plane of frequency should be seriously considered for verification even with very small sampling offsets. We also proposed a modified ANA-PHY-PCRAS for time-separated subchannels, which shows its robustness in verification whenever the local oscillator at the receiver may change over time.

Compared to the traditional challenge-response authentication scheme implemented at the upper layer, we conclude that ANA-PHY-PCRAS (or its modified version) can be further protected by the uncertainty from both the wireless channel and introduced artificial noise, which is of information-theoretic nature and could not be broken even with unlimited computational power.

%\bibliography{IEEEabrv,bib_wu}
% Generated by IEEEtran.bst, version: 1.13 (2008/09/30)

\end{document}